\documentclass[12pt]{article}%
\usepackage{amsmath}
\usepackage{amsfonts}
\usepackage{amssymb}
\usepackage{graphicx}%
\newtheorem{theorem}{Theorem}
\newtheorem{corollary}{Corollary}
\newenvironment{proof}[1][Proof]{\noindent\textbf{#1.} }{\ \rule{0.5em}{0.5em}}
\usepackage{hyperref}

\begin{document}

\author{Renan Cabrera, Traci Strohecker, Herschel Rabitz \\ 
{\small Department of Chemistry, Princeton University, Princeton, New Jersey 08544, USA}}

\title{ The Canonical Coset Decomposition of Unitary Matrices 
Through Householder Transformations }
\maketitle 
\abstract{ This paper reveals the relation between the canonical coset decomposition 
of unitary matrices and the corresponding decomposition via Householder reflections.
These results can be used to parametrize unitary matrices via Householder reflections.
}

\emph{\\Copyright (2010) American Institute of Physics. This article may be downloaded for personal use only. Any other use requires prior permission of the author and the American Institute of Physics.\\
The following article appeared in J. Math. Phys. 51, 082101 (2010); doi:10.1063/1.3466798 and may be found at \\
\url{http://jmp.aip.org/jmapaq/v51/i8/p082101_s1?isAuthorized=no}
 }

\section{Introduction}
The parametrization of unitary matrices is important in many fields of physics 
such as quantum computation \cite{PhysRevLett.73.58,ivanov:022323}, 
particle physics \cite{jarlskog2005quark} and
quantum iterferometry \cite{rowe1999representations}. Several methods for the parametrization of $SU(N)$
have been developed, including one that uses Euler angles 
\cite{tilma2002gea} to find explicit expressions for the Haar measure. 
Another form was given by Reck \emph{et al.}, who explicitly constructed arbitrary unitary operators in terms
of two-state pivots \cite{PhysRevLett.73.58}. 
Dita proposed a parametrization through diagonal unitary matrices interlaced with real orthogonal 
matrices \cite{dita2003factorization} and Rowe \emph{et al.} developed another method
studying the Wigner functions for $SU(3)$ \cite{rowe1999representations}, which were also
used in a formulation of coherent states for $SU(N)$ in the work of Nemoto\cite{nemoto2000generalized}. 
More recently, a technique resembling the canonical coset parametrization 
was presented by Jarlskog \cite{jarlskog2005recursive,fujii2006jarlskog}.  

Matrix decomposition through Householder reflections 
\cite{householder1958unitary} was recently proposed as an efficient method for synthesizing 
unitary operators\cite{ivanov:022323,ivanov:012335}. This technique can be implemented in certain
quantum systems using $N$ instead of $N^2$ steps as required in   
methods applying two-state pivots \cite{PhysRevLett.73.58}
following the logic of the Givens rotation \cite{givens1958computation}.
The canonical coset parametrization \cite{gilmore1973lgl,gilmore2008lgp} of unitary matrices
 is a general method that may be used to calculate explicit expressions of 
the Haar \cite{gilmore1973lgl,gilmore2008lgp,Boya2003401} and Bures measures
 \cite{1751-8121-40-37-010,akhtarshenas:012102}, which are known to be important in 
Bayesian quantum estimation \cite{slater1996qfb}. 
The main purpose of this paper is to establish the connection between the Householder decomposition 
and the canonical coset parametrization of unitary operators, which could have 
practical significance in a number of physical applications.

\section{Householder Decomposition}
A unitary operator can be seen as a basis of orthonormal vectors and the Householder 
decomposition of unitary operators serves to align the original basis in terms of 
a new basis. The Householder decomposition consists of a sequence of transformations, each of 
which is a reflection with respect to a hyperplane defined by the orthonormal 
vector $|n \rangle$, such that 
\begin{equation}
 R_{|n\rangle}  = \mathbf{1} - 2 |n \rangle \langle n |.
\end{equation}
As a proper reflection, $R_{|n \rangle}$ satisfies  $det(R_{|n \rangle})=-1$ and $ R_{|n\rangle} R_{|n\rangle} = \mathbf{1}$. Given a 
set of orthonormal basis elements $|e_k\rangle$, the unitary matrix $U(N)$ may be 
decomposed as a product of $N$ factors using Householder reflections that are designed
to sequentially align the columns of $U(N)$ along the orthonormal basis elements $|e_k\rangle$.

In order to illustrate the Householder decomposition, let $|W_1 \rangle$ be the first 
column of  $U(N)$ and $\phi_1$ the phase of the topmost complex 
component of $|W_1 \rangle$. The first Householder reflection is constructed as
\begin{equation}
R_{|u_1 \rangle} = \mathbf{1} - 2 \frac{1}{ \langle u_1| u_1 \rangle }  | u_1 \rangle \langle u_1 |,  
\end{equation}
where $| u_1 \rangle = |W_1 \rangle + e^{i\phi_1}| e_1 \rangle$, such that $R_{|u_1\rangle}|W_1\rangle = -|e_1\rangle$. For visualization, 
the case with two-dimensional real vectors is shown in Figure \ref{fig:householder}.
\begin{figure}[ht] 
\centering
\includegraphics[scale=0.7]{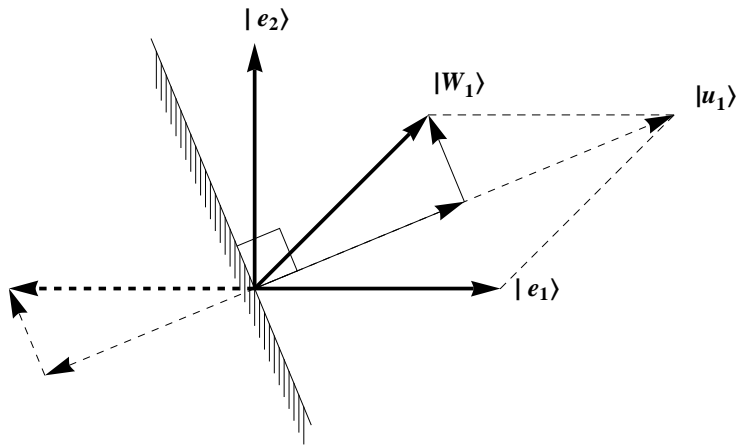}
\caption{ Graphical representation of a Householder reflection constructed 
for a real two-dimensional unit-length vector 
$| W_1 \rangle$ reflected on the plane perpendicular to $|u_1\rangle$. The net
result is a reflection that positions $| W_1 \rangle$  along the direction $-|e_1 \rangle$. 
An additional reflection can be performed to position  $| W_1 \rangle$  along  $|e_1 \rangle$.
}
\label{fig:householder}
\end{figure}
The most general definition of the vector $|u_1\rangle$ requires $|e_1 \rangle$
to be multiplied by the magnitude of the vector $|W_1 \rangle$, but as $U(N)$ is unitary,
then $|W_1 \rangle$ has unit norm. Although, the opposite sign choice  $e^{i\phi_1} \rightarrow -e^{i\phi_1}$
could be utilized, a positive sign is required in order to arrive at the canonical 
coset decomposition as shown in the proof of Theorem \ref{theorem-1} below. 
Moreover, the positive sign is usually chosen to obtain better numerical stability 
(see \cite{stoer2002introduction}, page 225).
The above procedure is repeated in recursion, with $|W_2 \rangle$ as the second 
column of $R_{|u_1 \rangle} U(N)$ (instead of $U(n)$) and  $| u_2 \rangle = |W_2 \rangle + e^{i\phi_2}| e_2 \rangle$. 
In this way, the resulting Householder decomposition becomes
\begin{equation}
 U(N) = R_{|u_1 \rangle}  R_{|u_2 \rangle}...  R_{|u_{N-1} \rangle} \,\, \mathcal{U}(1)^{N},
\end{equation} 
where $\mathcal{U}(1)^{N}$ is the diagonal matrix with element $e^{i\phi_l}$ at the $l$-th position
along the diagonal.

The construct of the canonical coset decomposition from the 
Householder decomposition is given by the following theorem 
\begin{theorem}
The factors of the canonical coset decomposition
\begin{equation}
U(N) = \frac{U(N)}{U(N-1)\otimes U(1)}\,
  \frac{U(N-1)}{U(N-2)\otimes U(1)}\,...\,
  \frac{U(2)}{U(1)\otimes U(1)}
  U(1)^{\otimes N}
\label{left-can-coset}, 
\end{equation}
can be expressed in terms of Householder reflections $R_{|u_k \rangle}$ as
\begin{eqnarray}
 \frac{U(N)}{U(N-1)\otimes  U(1)  } &=& R_{|u_1\rangle} R_{|e_1 \rangle} \\
\frac{U(N-1)}{U(N-2)\otimes  U(1)  } &=& R_{|u_2\rangle} R_{|e_2 \rangle} \\
&\vdots  & \nonumber \\
U(1)^{\otimes N} &=& R_{|e_1 \rangle} R_{|e_2 \rangle}... R_{|e_{N-1} \rangle}  {\mathcal{U}(1)}^{N}  
\end{eqnarray}
where the  corresponding Householder decomposition is 
\begin{equation}
 U(N) = R_{|u_1 \rangle}  R_{|u_2 \rangle}...  R_{|u_{N-1} \rangle} \,\, {\mathcal{U}(1)}^{N},
\end{equation} 
such that
  \begin{equation}
| u_1 \rangle = |W_1 \rangle + e^{i\phi_1}| e_1 \rangle,
\label{v1} 
\end{equation}  
with $|W_1\rangle$ being the first column of $U(N)$ and $\phi_1$ the phase of the first 
component of $|W_1 \rangle$. The remaining vectors $|u_k\rangle$ 
are calculated recursively. 
\label{theorem-1}
\end{theorem}

\begin{proof}
Let us define the normalized vector $|n_1 \rangle$ as
\begin{equation}
  |n_1 \rangle = \frac{1}{\sqrt{ \langle u_1 |u_1 \rangle }}|u_1 \rangle
\end{equation}
The first column of $U(N)$ can be decomposed into parallel and  perpendicular 
parts with respect to $|e_1 \rangle$ as
\begin{equation}
 | W_1  \rangle =  | W_{1\parallel}  \rangle +  | W_{1\perp}  \rangle =
   \rho e^{i\phi_1} |e_1 \rangle    +  | W_{1\perp}  \rangle,
\label{W1-def}
\end{equation}
where $\rho e^{i\phi_1}$ is the polar representation of the first component of $|W_1 \rangle$. The normalized 
vector $|n_1 \rangle$ can be expressed as
\begin{equation}
   |n_1 \rangle = \frac{1}{\sqrt{ \langle u_1 |u_1 \rangle }}\left( 
   |W_{1\perp}\rangle + ( \rho +1) e^{i\phi_1}| e_1 \rangle  \right).
\label{v-def}
\end{equation}
This allows writing  $|n_1 \rangle$ as
\begin{equation}
 |n_1 \rangle = |n_{\parallel} \rangle  +  |n_{\perp} \rangle  =  \gamma |e_1 \rangle +  |n_{\perp} \rangle, 
\label{n-gamma}
\end{equation}
with $\gamma =  \frac{ (\rho+1)e^{i\phi}  }{\sqrt{ \langle u_1 |u_1 \rangle }}  $ and $| n_{\perp} \rangle = \frac{ |W_{1\perp} \rangle}{\sqrt{ \langle u_1 |u_1 \rangle }} $.
 From (\ref{W1-def}) and (\ref{v-def}) we obtain
\begin{eqnarray}
  1 &=& \rho^2 +  \langle W_{1\perp} |W_{1\perp} \rangle \\
  \langle u_1 |u_1 \rangle &=&  \langle W_{1\perp} |W_{1\perp} \rangle + (\rho+1)^2,
\end{eqnarray}
which can be used to extract $  \langle u_1 |u_1 \rangle =2(1+\rho)$ and write $\gamma$ as
\begin{equation}
 \gamma = \sqrt{ \frac{(1+\rho)}{2}  }e^{i\phi}.
\label{gamma-rho} 
\end{equation} 
The product of the following Householder operators  can be expanded as
\begin{eqnarray}
 R_{|n\rangle} R_{|e_1 \rangle}  &=& 
 (\mathbf{1} - 2 |n \rangle \langle n | )
 (\mathbf{1} - 2 | e_1 \rangle \langle e_1 |) \\
 &=& \mathbf{1} - 2 |n_{\perp} \rangle \langle n_{\perp} | 
  - 2 \gamma | e_1 \rangle \langle n_{\perp} | + 2\gamma^{*} |  n_{\perp} \rangle \langle e_{1} | \nonumber \\
 & & +\,2( |\gamma|^2-1 )| e_1 \rangle \langle e_1 |. \nonumber
\label{step-1}
\end{eqnarray}
The identity matrix can be decomposed as $\mathbf{1} = \mathbf{1}_{\perp}+| e_1 \rangle \langle e_1 |$ 
and  introduced above to yield
\begin{equation}
 R_{|n\rangle} R_{|e_1 \rangle} = 
 \mathbf{1}_{\perp} - 2 |n_{\perp} \rangle \langle n_{\perp} | 
  - 2 \gamma | e_1 \rangle \langle n_{\perp} | + 2\gamma^{*} |  n_{\perp} \rangle \langle e_{1} |
 +( 2|\gamma|^2-1 )| e_1 \rangle \langle e_1 |,
\end{equation}
which can be decomposed into four parts expressed in block matrix form as
\begin{equation}
  R_{|n\rangle} R_{|e_1 \rangle} = 
   \begin{pmatrix} 2|\gamma|^2-1   & 
    - 2 \gamma  \langle n_{\perp} |  \\
   2\gamma^{*} |  n_{\perp} \rangle  &  \mathbf{1}_{\perp} - 2 |n_{\perp} \rangle \langle n_{\perp} |
   \end{pmatrix}.
\label{block-mat1}
\end{equation}
Next, the following variable change is applied
\begin{equation}
 |X \rangle = 2 \gamma^{*} | n_{\perp} \rangle,
\label{Xn} 
\end{equation}
and the variable $r$ is defined as the magnitude of $|X\rangle$
\begin{equation}
 r = \sqrt{ \langle X| X\rangle }.
\end{equation}
From (\ref{n-gamma}) we see that
\begin{equation}
 r^2 = 4 |\gamma|^2 \langle n_{\perp}  | n_{\perp} \rangle 
     = 4 |\gamma|^2 (1-|\gamma|^2), 
\end{equation}
and with some algebra this leads to 
 \begin{equation} \sqrt{1-r^2} = 2 |\gamma|^2-1.  
\label{r-def}
 \end{equation} 
This construction is consistent only if $ 1 \ge 2 |\gamma|^2-1 \ge 0 $,
which is seen to be true by inspecting (\ref{gamma-rho}). Moreover, 
(\ref{gamma-rho}) implies that $  2 |\gamma|^2-1 = \rho $.

The block matrix (\ref{block-mat1}) can be written in the form of
a canonical coset as follows
\begin{eqnarray}
  R_{|n\rangle} R_{|e_1 \rangle} &=& 
   \begin{pmatrix}  \sqrt{1-r^2}  & 
   -  \langle X | \\
     | X \rangle &   \mathbf{1}_{\perp} - \frac{1-\sqrt{1-r^2}}{r^2}  |X \rangle \langle X |
   \end{pmatrix} \\
  &=& 
  \begin{pmatrix} \sqrt{1-r^2}   & 
   -  \langle X | \\
     | X \rangle &  \sqrt{ \mathbf{1}_{\perp} - |X \rangle \langle X | }
   \end{pmatrix}\\
   &=&  \frac{U(N)}{U(N-1)\otimes U(1)} ,
\end{eqnarray}
which is the form given by Gilmore\cite{gilmore1973lgl,gilmore2008lgp}.

The remainder of the cosets can be found recursively with the 
insertion of reflections of the form $R_{|e_k \rangle}$, which can be constructed by
replacing the k-th element of the identity matrix with $-1$.  
There is flexibility in the sequential choice of the reflections $R_{|e_k \rangle}$ 
because they commute with each other $[R_{|e_j \rangle} , R_{|e_k \rangle} ] = 0$.
Additionally, since $R_{|u_k \rangle}$ is a block matrix, the following identity can be verified  
\begin{equation}
 [R_{|e_j \rangle} , R_{|u_k \rangle} ] = 0 \,\,\,\, \text{for}\,\, k>j.
\end{equation}
Thus, the theorem is proved.
\end{proof}

\begin{corollary}
The factors of the \emph{reversed} canonical coset decomposition,  
\begin{equation}
U(N) =  U(1)^{\otimes N}
 \frac{U(2)}{U(1)\otimes U(1)}\,\, ... \frac{U(N-1)}{U(N-2)\otimes U(1)}\,\,
 \frac{U(N)}{U(N-1)\otimes U(1)} 
\label{rev-coset}, 
\end{equation}
can be expressed in terms of the reversed Householder reflections $R_{|u_k \rangle}$ as
\begin{eqnarray}
 \frac{U(N)}{U(N-1)\otimes  U(1)  } &=& R_{|e_1 \rangle} R_{|u_1\rangle} \\
\frac{U(N-1)}{U(N-2)\otimes  U(1)  } &=& R_{|e_2 \rangle} R_{|u_2\rangle} \\
&\vdots  & \nonumber \\
U(1)^{\otimes N} &=&  {\mathcal{U}(1)}^{N}R_{|e_{N-1} \rangle}... R_{|e_2 \rangle}  R_{|e_1 \rangle}  
\end{eqnarray}
where, the  corresponding Householder decomposition is 
\begin{equation} 
 U(N) =    {\mathcal{U}(1)}^{N}  R_{|u_{N-1} \rangle} ...  
 R_{|u_2 \rangle}  R_{|u_1 \rangle},
\end{equation} 
such that
  \begin{equation}
\langle u_1 | =  \langle W_1 | + e^{i\phi_1} \langle e_1 |,
\label{rev-v1} 
\end{equation}  
with $\langle W_1|$ being the first \emph{row} of $W$ and $\phi_1$ the phase of the first 
component of $\langle W_1 |$. The remainder of the vectors $ \langle u_k |$ 
are calculated recursively. 
\end{corollary}

One important consequence of Theorem \ref{theorem-1}  is the opportunity 
it affords to  parametrize the Householder reflections. This can be accomplished with the use of
(\ref{Xn}), as described in the following corollary.
\begin{corollary}
The normal Householder vector $|n\rangle$ can be
parametrized in terms of the canonical coset vector
$|X \rangle$ as
\begin{equation}
 | n \rangle = \gamma | e_1 \rangle +   \frac{1}{2 \gamma^{*}} |X \rangle,
\label{nX}
\end{equation}
where
 \begin{equation} 
 |\gamma| =  \frac{\sqrt{1-  \langle X|X \rangle} + 1}{2}, 
 \end{equation}
with the phase of $\gamma$ extracted from the phase of the first diagonal element of $ {\mathcal{U}(1)}^{N} $.
\end{corollary}

The parametrization of unitary matrices is essential in order to calculate the Haar metric 
and Haar measure \cite{Boya2003401} of unitary matrices, and for this purpose
 using Householder reflections is more efficient than other alternatives such as the parametrization in terms 
of Euler angles\cite{tilma2002gea}. 
In the same spirit, the parametrization of the Householder 
normal vectors $|n \rangle$ (\ref{nX}) can be exploited to generate random unitary matrices with 
the Haar measure as suggested by Ivanov \cite{ivanov-optimal}. 
This can be accomplished by homogeneously distributing $|X \rangle$ within Euclidean balls \cite{1751-8121-42-44-445302}
while uniformly distributing the phase of $\gamma$ in the range $ arg(\gamma) \in [-\pi,\pi]$.
For example, the generation of  random unitary matrices $U(3)$ with the Haar 
measure requires the following parametrization
\begin{equation}
U(3) =  (\mathbf{1} - 2 |n_1 \rangle \langle n_1|)(\mathbf{1} - 2 |n_2 \rangle \langle n_2|)
diag(e^{i\phi_1}, e^{i\phi_2}, e^{i\phi_3}),
\end{equation}
where the phases $\phi_k$ are drawn from a uniform distribution over the 
range $[-\pi,\pi]$ and where $|n_1 \rangle$ and  $|n_2\rangle$ are parametrized within balls $B^4$ and $B^2$, 
respectively.

As an illustrative example for the calculation of the canonical coset decomposition, consider 
\begin{equation}
  U_0 = \begin{pmatrix}
   \frac{i}{\sqrt{2}} & \frac{i}{\sqrt{2}} & 0  \\
  -\frac{i}{2} & \frac{i}{2} & \frac{i}{\sqrt{2}} \\
  -\frac{1}{2} & \frac{1}{2} &-\frac{1}{\sqrt{2}}
   \end{pmatrix}.
\label{W-example}
\end{equation}
The Householder decomposition is $U_0 = | R \rangle_{|u_1\rangle}  | R \rangle_{|u_2\rangle} \mathcal{U}(1)^3  $, where
\begin{eqnarray}
  | R \rangle_{|u_1\rangle} &=& \begin{pmatrix}
   -\frac{1}{\sqrt{2}} & \frac{1}{2} &\frac{i}{2} \\
  \frac{1}{2} & \frac{(2+\sqrt{2})}{4} & -\frac{i}{4+2\sqrt{2}} \\
  -\frac{i}{2} & \frac{i}{4 + 2\sqrt{2}} & \frac{2 + 2\sqrt{2}}{4}
   \end{pmatrix} \\
  | R \rangle_{|u_2\rangle} &=& \begin{pmatrix}
  1 &0 & 0 \\ 0 & -\frac{1}{\sqrt{2}}   & -\frac{i(1+\sqrt{2})}{2+\sqrt{2} }\\
  0 & \frac{i(1+\sqrt{2})}{2+\sqrt{2} } & \frac{1}{\sqrt{2}}
  \end{pmatrix}\\
 {\mathcal{U}(1)}^3  &=& \begin{pmatrix}
 -i & 0  & 0 \\
 0  & -i & 0 \\
 0  & 0  & -1
\end{pmatrix}.
\end{eqnarray}
This leads to the corresponding canonical coset decomposition
\begin{eqnarray}
  \frac{U(3)}{U(2)\otimes U(1)} &=& \begin{pmatrix}
   \frac{1}{\sqrt{2}} & \frac{1}{2} &\frac{i}{2} \\
  -\frac{1}{2} & \frac{(2+\sqrt{2})}{4} & -\frac{i}{4+2\sqrt{2}} \\
  \frac{i}{2} & \frac{i}{4 + 2\sqrt{2}} & \frac{2 + 2\sqrt{2}}{4}
   \end{pmatrix} \\
  \frac{U(2)}{U(1)\otimes U(1)}  &=& \begin{pmatrix}
  1 &0 & 0 \\ 0 & \frac{1}{\sqrt{2}}   & -\frac{i(1+\sqrt{2})}{2+\sqrt{2} }\\
  0 & -\frac{i(1+\sqrt{2})}{2+\sqrt{2} } & \frac{1}{\sqrt{2}}
  \end{pmatrix}\\
 U(1)^3  &=& \begin{pmatrix}
 i & 0  & 0 \\
 0  & i & 0 \\
 0  & 0  & -1
\end{pmatrix}.
\end{eqnarray}

The reversed canonical coset decomposition (\ref{rev-coset}) can be obtained using the 
reversed Householder reflections with the following results
\begin{eqnarray}
  \frac{U(3)}{U(2)\otimes U(1)} &=& \begin{pmatrix}
   \frac{1}{\sqrt{2}} & \frac{1}{2} & 0  \\
  -\frac{1}{\sqrt{2}} & \frac{1}{\sqrt{2}} & 0 \\
  0 & 0 & 1
   \end{pmatrix} \\
  \frac{U(2)}{U(1)\otimes U(1)}  &=& \begin{pmatrix}
  1 &0 & 0 \\ 0 & \frac{1}{\sqrt{2}}   & \frac{1 }{\sqrt{2}}\\
  0 & -\frac{1}{\sqrt{2}} & \frac{1}{\sqrt{2}}
  \end{pmatrix}\\
 U(1)^3  &=& \begin{pmatrix}
 i & 0  & 0 \\
 0  & i & 0 \\
 0  & 0  & -1
\end{pmatrix}.
\end{eqnarray}

\section{Conclusions}
We formally proved the connection between the Householder
and canonical coset decompositions of unitary matrices. This result 
allows for performing the canonical coset decomposition of unitary matrices 
very efficiently. Furthermore, the Householder decomposition is now fully 
parametrized in terms of the canonical coset vectors $|X \rangle$ (\ref{nX}).

\section{Acknowledgment}
The authors acknowledge the support from the DOE and ARO.

\section*{Appendix }
The following exponential involving the complex vector $B$ can be written 
in terms of the Cartesian coordinates $x^j$ as
\begin{equation}
   \exp{ \begin{pmatrix} \mathbf{0} & B \\ -B^\dagger &0 \end{pmatrix}  } 
= \begin{pmatrix} [\mathbf{1} - XX^\dagger]^{1/2} & X \\
   -X^\dagger &  [1 - X^\dagger X]^{1/2}  \end{pmatrix}.
\label{exp-coset}
\end{equation}
such that 
\begin{equation}
 X =  \frac{\sin \sqrt{B^\dagger}B}{\sqrt{B^\dagger}B  } B =
 \begin{pmatrix}
  x^1 + i x^2 \\ x^3 +i x^4\\ \vdots \\ x^{2N-3} + i x^{2N-2}
 \end{pmatrix}
\end{equation}
The Cartesian coordinates range inside an even ball $B^{2k}$, where the radial
coordinate is $r^2 = X^\dagger X $. This exponential is important because it 
provides a parametrization of the coset $ \frac{U(N)}{U(N-1)\otimes U(1)}$ as a $ N \times N$ matrix, which
can be used to parametrize the unitary operator $\Omega$ as
\begin{equation}
\Omega \in  \frac{U(N)}{U(1)^{\otimes N}} =  
 \frac{U(N)}{U(N-1)\otimes U(1)} \,\, \frac{U(N-1)}{U(N-2)\otimes U(1)}...\frac{U(2)}{U(1)\otimes U(1)}.
\label{main-coset} 
\end{equation}
In the current paper we are interested in the representation of $U(3)$, which are generated 
from two cosets. The left coset is
\begin{equation}
 \frac{U(2)}{U(1)\otimes U(1)} = 
\begin{pmatrix}
1 & 0 & 0  \\
0 & \sqrt{1- (x^1)^2 - (x^2)^2 } & -x^1 +i x^2 \\
0 & x^1 +i x^2 & \sqrt{1- (x^1)^2 - (x^2)^2 }
\end{pmatrix}, 
\end{equation}
with the variables inside a disk $ (x^1)^2 + (x^2)^2 \le 1$. The right coset is 
\begin{equation}
 \frac{U(3)}{U(2)\otimes U(1)} = 
\begin{pmatrix}
\sqrt{1 - \xi^2} & -x^5 + i x^6 & -x^3 + i x^4  \\
x^5 + i x^6  & V_{22} & V_{23} \\
x^3 + i x^4  & V_{23}^* & V_{33}
\end{pmatrix}, 
\end{equation}
with 
\begin{eqnarray}
V_{22} &=&  \frac{ (x^3)^2 + (x^4)^2 + \sqrt{1-\xi^2}( { (x^5)^2 + (x^6)^2   )  }}{ \xi^2 } \\
V_{23} &=&  \frac{\sqrt{1-\xi^2} - 1}{\xi^2} (x^3-ix^4)(x^5 + i x^6) \\
V_{33} &=&  \frac{\sqrt{1-\xi^2}}{\xi^2}(x^3 + x^4) + \frac{1}{\xi^2}(x^5 + x^6)\\
\xi^2  &=&  (x^3)^2 + (x^4)^2  +  (x^5)^2  +  (x^6)^2
\end{eqnarray}
and the variables in the following range
\begin{equation}
 (x^3)^2 + (x^4)^2 +  (x^5)^2 + (x^6)^2    \le 1 
\end{equation}

\bibliographystyle{unsrt}


\end{document}